\documentclass[11pt,a4paper,reqno]{amsart}
\usepackage{multirow}
\usepackage[centertags]{amsmath}
\usepackage{amsfonts}
\usepackage{amssymb}
\usepackage{amsthm}
\usepackage{newlfont}
\usepackage{a4}
\usepackage{enumerate}
\usepackage[pdftex]{graphicx,color}
\theoremstyle{definition}
\newtheorem{defn}{Definition}[section]
\newtheorem{thm}[defn]{Theorem}

\newtheorem{tvr}[defn]{Proposition}
\newtheorem{cor}[defn]{Corollary}
\theoremstyle{remark}
\newtheorem{example}{Example}[section]
\newcommand{\id}{\mathfrak{1}}
\usepackage{bbm}

\newlength{\defbaselineskip}
\newcommand{\setlinespacing}[1]%
           {\setlength{\baselineskip}{#1 \defbaselineskip}}

\renewcommand{\i}{\mathrm{i}}
\newcommand{\map}{\rightarrow}

\newcommand{\q}{\quad}

\renewcommand{\epsilon}{\varepsilon}

\newcommand{\ep}{\varepsilon}
\newcommand{\la}{\lambda}
\newcommand{\al}{\alpha}
\newcommand{\om}{\omega}
\renewcommand{\rho}{\varrho}
\renewcommand{\phi}{\varphi}

\newcommand{\R}{{\mathbb{R}}}

\newcommand{\N}{{\mathbb N}}

\newcommand{\Com}{{\mathbb C}}
\newcommand{\Z}{\mathbb{Z}}
\newcommand{\C}{\mathbb{C}}

\newcommand{\set}[2]{\left\{#1 \, |\, #2 \right\}}
\newcommand{\setb}[2]{\left\{#1 \, \mid\, #2 \right\}}
\newcommand{\setm}[2]{\left\{#1 \,\, \big|\,\, #2 \right\}}
\newcommand{\abs}[1]{\left\vert#1\right\vert}

\newcommand{\sca}[2]{\langle #1,\, #2\rangle}
\newcommand{\comb}[2]{\begin{pmatrix}
     #1\\
     #2
  \end{pmatrix}}

\begin{document}

\title[Torus discretization II]
{On Discretization of Tori of Compact Simple Lie Groups II}

\author{Ji\v{r}\'{i} Hrivn\'{a}k$^{1}$}
\author{Lenka Motlochov\'{a}$^{3}$}
\author{Ji\v{r}\'{i} Patera$^{2,3,4}$}

\date{\today}
\begin{abstract}\
The discrete orthogonality of special function families, called $C$- and $S$-functions, which are derived from the characters of compact simple Lie groups, is described in \cite{HP}. Here, the results of \cite{HP} are extended to two additional recently discovered families of special functions, called $S^s-$ and $S^l-$functions. The main result is an explicit description of their pairwise discrete orthogonality within each family, when the functions are sampled on finite fragments $F^s_M$ and $F^l_M$ of a lattice in any dimension $n\geq2$ and of any density controlled by $M$, and of the symmetry of the weight lattice of any compact simple Lie group with two different lengths of roots.
\end{abstract}

\maketitle
\noindent
$^1$ Department of physics, Faculty of nuclear sciences and physical engineering, Czech Technical University in Prague, B\v{r}ehov\'a~7, CZ-115 19 Prague, Czech Republic\\
$^2$ Centre de recherches math\'ematiques, Universit\'e de Montr\'eal, C.~P.~6128 -- Centre ville, Montr\'eal, H3C\,3J7, Qu\'ebec, Canada\\
$^3$ D\'epartement de Math\'ematiques et de Statistique, Universit\'e de Montr\'eal, Qu\'ebec, Canada;\\
$^4$ MIND Research Institute, 3631 S. Harbor Blvd., Suite 200, Santa Ana, CA 92704, USA
\vspace{10pt}

\noindent
E-mail: jiri.hrivnak@fjfi.cvut.cz, motlochova@dms.umontreal.ca, patera@crm.umontreal.ca


\section{Introduction}
This paper focuses on the Fourier transform of data sampled on lattices of any dimension and any symmetry \cite{MP2,P}. The main problem is to find families of expansion functions that are complete in their space and orthogonal over finite fragments of the lattices. Generality of results is possible because the expansion function is built using properties that are uniformly valid over the series of semisimple Lie groups. Results of \cite{HP} on the discrete orthogonality of $C-$ and $S-$functions of compact simple Lie groups are extended to the recently discovered families of $S^s-$ and $S^l-$functions. The new families of functions add new possibilities of transforms for the same data.

Uniform discretization of tori of all semisimple Lie groups became possible after the classification of conjugacy classes of elements of finite order in compact simple Lie groups \cite{Kac}. This was accomplished in \cite{MP3,MP1} for $C-$functions and extended to $S-$functions in \cite{MP2}. Functions of  $C-$ and $S-$families are ingredients of irreducible characters of representations. They are uniformly defined for all semisimple Lie groups. Discretization here refers to their orthogonality when sampled on a fraction of a lattice $F_M$ in the fundamental region $F$ of the corresponding Lie group and summed up over all lattice points in $F_M$. This lattice is necessarily isomorphic to the weight lattice of the underlying Lie group, but its density is controlled by the choice of $M\in\N$.

$C-$functions are Weyl group invariant constituents of characters of irreducible representations. They are well known, even if unfrequently used \cite{KP1}. $S-$functions appear in the Weyl character formula. They are skew-invariant with respect to the Weyl group \cite{KP2}. In the new families of $S^s-$ and $S^l-$functions, the Weyl group acts differently when reflections are, with respect to hyperplanes, orthogonal to short and long roots of the Lie group. The functions are `half invariant and half skew-invariant' under the action of the Weyl group.

The key point of the discretization of $S^s-$ and $S^l-$functions lies in finding appropriate subsets $F_M^s\subset F_M$ and $F_M^l\subset F_M$, which play the role of sampling points of a given data. The solution involves determining the sets of weights $\Lambda_M^s$ and $\Lambda_M^l$, which label the discretely orthogonal $S^s-$ and $S^l-$functions over the sets $F_M^s$ and $F_M^l$. In order to verify the completeness of the found sets of functions, the last step involves comparing the number of points in $F_M^s$, $F_M^l$ to the number of weights in $\Lambda_M^s$, $\Lambda_M^l$.

The pertinent standard properties of affine Weyl groups and their dual versions are recalled in section~2. Two types of sign homomorphisms and the corresponding fundamental domains are defined in section~3. The $S^s-$ and $S^l-$functions and their behavior on the given discrete grids are studied in section~4. In section~5, the number of points in $F_M^s$, $F_M^l$ are shown to be equal to the number of weights in $\Lambda_M^s$, $\Lambda_M^l$. Explicit formulas for these numbers are also given. Section~6 contains the detailed description of the discrete orthogonality and discrete transforms of $S^s-$ and $S^l-$functions. Comments and follow-up questions are in the last section.

\
\section{Pertinent properties of affine Weyl groups}

\subsection{Roots and reflections}\

We use the notation established in \cite{HP}. Recall that, to the Lie algebra of the compact simple Lie group $G$ of rank $n$, corresponds the set of simple roots $\Delta=\{\al_1,\dots,\al_n\}$ \cite{BB,H2,VO}. The set $\Delta$ spans the Euclidean space $\R^n$, with the scalar product denoted by $\sca{\,}{\,}$. We consider here only simple algebras with two different lengths of roots, namely $B_n,\, n\geq 3$, $C_n,\, n\geq 2$, $G_2$ and $F_4$. For these algebras, the set of simple roots consists of short simple roots $\Delta_s$ and long simple roots $\Delta_l$. Thus, we have the disjoint decomposition
\begin{equation}\label{sl}
\Delta=\Delta_s\cup\Delta_l.
\end{equation}
We then use the following well-known objects related to the set $\Delta$:
\begin{itemize}
\item
The marks $m_1,\dots,m_n$ of the highest root $\xi\equiv -\al_0=m_1\al_1+\dots+m_n\al_n$.
\item
The Coxeter number $m=1+m_1+\dots+m_n$ of $G$.
\item
The Cartan matrix $C$ and its determinant
\begin{equation}\label{Center}
 c=\det C.
\end{equation}
\item
The root lattice $Q=\Z\al_1+\dots+\Z\al_n $.

\item
The $\Z$-dual lattice to $Q$,
\begin{equation*}
 P^{\vee}=\set{\om^{\vee}\in \R^n}{\sca{\om^{\vee}}{\al}\in\Z,\, \forall \al \in \Delta}=\Z \om_1^{\vee}+\dots +\Z \om_n^{\vee}\,.
 \end{equation*}
\item
The dual root lattice $Q^{\vee}=\Z \al_1^{\vee}+\dots +\Z \al^{\vee}_n$, where $\al^{\vee}_i=2\al_i/\sca{\al_i}{\al_i}$.
\item
The dual marks $m^{\vee}_1, \dots ,m^{\vee}_n$ of the highest dual root $\eta\equiv -\al_0^{\vee}= m_1^{\vee}\al_1^{\vee} + \dots + m_n^{\vee} \al_n^{\vee}$. The marks and the dual marks are summarized in Table 1 in \cite{HP}.
\item The $\Z$-dual lattice to $Q^\vee$
\begin{equation*}
 P=\set{\om\in \R^n}{\sca{\om}{\al^{\vee}}\in\Z,\, \forall \al^{\vee} \in Q^\vee}=\Z \om_1+\dots +\Z \om_n.
\end{equation*}
\end{itemize}

Recall that $n$ reflections $r_\al$, $\al\in\Delta$ in $(n-1)$-dimensional `mirrors' orthogonal to simple roots intersecting at the origin are given explicitly by
\begin{equation*}
r_{\al}a=a-\frac{2\sca{a}{\al} }{\sca{\al}{\al}}\al\,,
\qquad a\in\R^n\,.
\end{equation*}
The affine reflection $r_0$ with respect to the highest root $\xi$ is given by
\begin{equation*}
r_0 a=r_\xi a + \frac{2\xi}{\sca{\xi}{\xi}}\,,\qquad
r_{\xi}a=a-\frac{2\sca{a}{\xi} }{\sca{\xi}{\xi}}\xi\,,\qquad a\in\R^n\,.
\end{equation*}

We denote the set of reflections $r_1\equiv r_{\al_1}, \, \dots, r_n\equiv r_{\al_n}$, together with the affine reflection $r_0$, by $R$, i.e.
\begin{equation}\label{R}
R=\{ r_0,r_1,\dots,r_n \}.
\end{equation}
Analogously to (\ref{sl}), we divide the reflections of $R$ into two subsets
\begin{align*}
R^s&=\set{r_\al}{\al\in\Delta_s }\\
R^l&=\set{r_\al}{\al\in\Delta_l }\cup\{r_0\}.
\end{align*}
We then obtain the disjoint decomposition
\begin{equation}\label{slaff}
R=R^s\cup R^l.
\end{equation}
We can also call the sums of marks corresponding to the long or short roots, i.e. the numbers
\begin{align*}
m^s&=\sum_{\al_i\in \Delta_s}m_i \\
m^l&=\sum_{\al_i\in \Delta_l}m_i+1,
\end{align*}
the short and the long Coxeter numbers. Their sum gives the Coxeter number $m=m^s+m^l$.

The dual affine reflection $r_0^{\vee}$, with respect to the dual highest root $\eta$, is given by
\begin{equation*}
r_0^{\vee} a=r_{\eta} a + \frac{2\eta}{\sca{\eta}{\eta}}, \q r_{\eta}a=a-\frac{2\sca{a}{\eta} }{\sca{\eta}{\eta}}\eta,\q a\in\R^n.
\end{equation*}
We denote the set of reflections $r^\vee_1\equiv r_{\al_1}, \, \dots, r^\vee_n\equiv r_{\al_n}$, together with the dual affine reflection $r^\vee_0$, by $R^\vee$, i.e.
\begin{equation}\label{Rd}
R^{\vee}=\{ r_0^{\vee},r^\vee_1,\dots,r^\vee_n \}.
\end{equation}
Analogously to (\ref{sl}), (\ref{slaff}), we divide the reflections of $R^\vee$ into two subsets
\begin{align*}
R^{ s\vee}&=\set{r_\al}{\al\in\Delta_s }\cup\{r^\vee_0\}\\
R^{ l\vee}&=\set{r_\al}{\al\in\Delta_l }.
\end{align*}
The disjoint decomposition of $R^\vee$ is then
\begin{equation}\label{slaffdual}
R^\vee=R^{s\vee}\cup R^{l\vee}.
\end{equation}
We can also call the sums of the dual marks corresponding to generators from $R^{s\vee}$, $R^{l\vee}$ i.e. the numbers
\begin{align*}
m^{s\vee}&=\sum_{\al_i\in \Delta_s}m^\vee_i+1 \\
m^{l\vee}&=\sum_{\al_i\in \Delta_l}m^\vee_i,
\end{align*}
the short and the long dual Coxeter numbers. Again, their sum gives the Coxeter number $m=m^{s\vee}+m^{l\vee}$.
Direct calculation of these numbers yields the following crucial result:
\begin{tvr}\label{msml}
For the numbers $m^s$, $m^l$ and $m^{s\vee}$, $m^{l\vee}$, it holds that
\begin{equation}
m^s=m^{s\vee},\q m^l=m^{l\vee}.
\end{equation}
\end{tvr}

The explicit form of decompositions \eqref{slaff} and \eqref{slaffdual} of sets $R$ and $R^\vee$ as well as Coxeter numbers $m^s$, $m^l$ is given in Table \ref{tabdec}.
{
\begin{table}
\begin{tabular}{|c||c|c|c|c|c|c|}
\hline
Type & $R^s$  & $R^l$ & $R^{s\vee}$ & $R^{l\vee}$ &  $m^s$ & $m^l$ \\
\hline\hline
$B_n\ (n\geq3)$ & $r_n$ & $r_0,\, r_1, \dots, r_{n-1}$ & $r^\vee_0,\, r^\vee_n$ & $r^\vee_1, \dots, r^\vee_{n-1}$ & $2$ & $2n-2$ \\ \hline
$C_n\ (n\geq2)$ & $r_1, \dots, r_{n-1}$ & $r_0,\,r_n$ & $r^\vee_0,r^\vee_1, \dots, r^\vee_{n-1}$ & $r^\vee_n$ & $2n-2$ & $2$ \\ \hline
$G_2$ & $r_2$ & $r_0,\,r_1$ & $r^\vee_0,\, r^\vee_2$ & $r^\vee_1$ & $3$ & $3$ \\ \hline
$F_4$ & $r_3,r_4$ & $r_0,\,r_1,\, r_2$ & $r^\vee_0,\, r^\vee_3, r^\vee_4$ & $r^\vee_1,\, r^\vee_2$ & $6$ & $6$\\ \hline
\end{tabular}
\medskip
\caption{ The decomposition of the sets of generators $R$, $R^\vee$ and the Coxeter numbers $m^s$, $m^l$. Numbering of the simple roots is standard (see e.g. Figure 1 in \cite{HP}).}\label{tabdec}
\end{table}
}

\subsection{Weyl group and affine Weyl group}\

Weyl group $W$ is generated by $n$ reflections $r_\al$, $\al\in\Delta$. Applying the action of $W$ on the set of simple roots $\Delta$, we obtain the entire root system $W\Delta$. The root system $W\Delta$ contains two subsystems $W\Delta_s$ and $W\Delta_l$, i.e. we have the disjoint decomposition
\begin{equation}\label{rsl}
W\Delta=W\Delta_s\cup W\Delta_l.
\end{equation}

The set of $n+1$ generators $R$ generates the affine Weyl group $W^{\mathrm{aff}}$.
The affine Weyl group $W^{\mathrm{aff}}$ can be viewed as the semidirect product of the Abelian group of translations $Q^\vee$ and of the Weyl group~$W$:
\begin{equation}\label{direct}
 W^{\mathrm{aff}}= Q^\vee \rtimes W.
\end{equation}
Thus, for any $w^{\mathrm{aff}}\in W^{\mathrm{aff}}$, there exist a unique $w\in W$ and a unique shift $T(q^{\vee})$ such that $w^{\mathrm{aff}}=T(q^{\vee})w$.
The retraction homomorphism $\psi:{W^{\mathrm{aff}}}\map W $ for $w^{\mathrm{aff}}\in {W^{\mathrm{aff}}}$ is given by
\begin{equation}\label{ret}
\psi(w^{\mathrm{aff}})=\psi(T(q^{\vee})w)=w.
\end{equation}

The fundamental domain $F$ of $W^{\mathrm{aff}}$, which consists of precisely one point of each $W^{\mathrm{aff}}$-orbit, is the convex hull of the points $\left\{ 0, \frac{\om^{\vee}_1}{m_1},\dots,\frac{\om^{\vee}_n}{m_n} \right\}$. Considering $n+1$ real parameters $y_0,\dots, y_n\geq 0$, we have
\begin{align}
F &=\setb{y_1\om^{\vee}_1+\dots+y_n\om^{\vee}_n}{y_0+y_1 m_1+\dots+y_n m_n=1  }. \label{deffun}
\end{align}
Recall that the stabilizer
\begin{equation}
\mathrm{Stab}_{W^{\mathrm{aff}}}(a) = \setb{w^{\mathrm{aff}}\in W^{\mathrm{aff}}}{w^{\mathrm{aff}}a=a}
\end{equation}
of a point $a=y_1\om^{\vee}_1+\dots+y_n\om^{\vee}_n\in F$ is trivial, $\mathrm{Stab}_{W^{\mathrm{aff}}}(a)=1$ if the point $a$ is in the interior of $F$, $a\in \mathrm{int}(F)$. Otherwise the group $\mathrm{Stab}_{W^{\mathrm{aff}}}(a)$
is generated by such $r_i$ for which $y_i=0$, $i=0,\dots,n$.

Considering the standard action of $W$ on the torus $\R^n/Q^{\vee}$, we denote for $x\in \R^n/Q^{\vee}$ the isotropy group and its order by
\begin{equation*}
\mathrm{Stab} (x)=\set{w\in W}{wx=x},\q h_x\equiv |\mathrm{Stab} (x)|
\end{equation*}
and denote the orbit and its order by
 \begin{equation*}
W x=\set{wx\in \R^n/Q^{\vee} }{w\in W},\q \ep(x)\equiv |Wx|.
\end{equation*}
Then we have
\begin{equation}\label{ep}
\ep(x)=\frac{|W|}{h_x}.
\end{equation}
Recall the following three properties from Proposition 2.2 in \cite{HP} of the action of $W$ on the torus $\R^n/Q^{\vee}$:
\begin{enumerate}
\item For any $x\in \R^n/Q^{\vee}$, there exists $x'\in F \cap \R^n/Q^{\vee} $ and $w\in W$ such that
\begin{equation}\label{rfun1}
 x=wx'.
\end{equation}

\item If $x,x'\in F \cap \R^n/Q^{\vee} $ and $x'=wx$, $w\in W$, then
\begin{equation}\label{rfun2}
 x'=x=wx.
\end{equation}
\item If $x\in F \cap \R^n/Q^{\vee} $, i.e. $x=a+Q^{\vee}$, $a\in F$, then $\psi (\mathrm{Stab}_{W^{\mathrm{aff}}}(a))=\mathrm{Stab}(x)$ and
\begin{equation}\label{rfunstab}
\mathrm{Stab} (x) \cong \mathrm{Stab}_{W^{\mathrm{aff}}}(a).
\end{equation}
\end{enumerate}

\subsection{Dual affine Weyl group}\

The dual affine Weyl group $\widehat{W}^{\mathrm{aff}}$ is generated by the set $R^\vee$.
Moreover, $\widehat{W}^{\mathrm{aff}}$ is a semidirect product of the group of shifts $Q$ and the Weyl group $W$
\begin{equation}\label{directd}
 \widehat{W}^{\mathrm{aff}}= Q \rtimes W.
\end{equation}
Thus, for any $w^{\mathrm{aff}}\in\widehat{W}^{\mathrm{aff}}$, there exist a unique $w\in W$ and a unique shift $T(q)$ such that $w^{\mathrm{aff}}=T(q)w$.
The dual retraction homomorphism $\widehat\psi:\widehat{W}^{\mathrm{aff}}\map W $ for $w^{\mathrm{aff}}\in \widehat{W}^{\mathrm{aff}}$ is given by
\begin{equation}\label{retd}
\widehat\psi(w^{\mathrm{aff}})=\widehat\psi(T(q)w)=w.
\end{equation}
The dual fundamental domain $F^\vee$ of $\widehat{W}^{\mathrm{aff}}$ is the convex hull of vertices $\left\{ 0, \frac{\om_1}{m^{\vee}_1},\dots,\frac{\om_n}{m^{\vee}_n} \right\}$. Considering $n+1$ real parameters $z_0,\dots, z_n\geq 0$, we have
\begin{align}
F^\vee &=\setb{z_1\om_1+\dots+z_n\om_n}{z_0+z_1 m_1^{\vee}+\dots+z_n m^{\vee}_n=1  }.\label{deffund} 
\end{align}

Consider the point $a=z_1\om_1+\dots+z_n\om_n\in F^\vee$ such that $z_0+z_1 m_1^{\vee}+\dots+z_n m^{\vee}_n=1$. The isotropy group
\begin{equation}\label{stabdual}
 \mathrm{Stab}_{\widehat{W}^{\mathrm{aff}}}(a) = \setm{w^{\mathrm{aff}}\in \widehat{W}^{\mathrm{aff}}}{w^{\mathrm{aff}}a=a}
\end{equation}
of point $a$ is trivial, $ \mathrm{Stab}_{\widehat{W}^{\mathrm{aff}}}(a)=1$, if $a\in \mathrm{int}(F^\vee)$, i.e. all $z_i>0$, $i=0,\dots,n$. Otherwise the group $\mathrm{Stab}_{\widehat{W}^{\mathrm{aff}}}(a)$
is generated by such $r^{\vee}_i$ for which $z_i=0$, $i=0,\dots,n$.

Recall from \cite{HP} that, for an arbitrary $M\in\N$, the grid $\Lambda_M$ is defined as cosets from the $W-$invariant group $P/MQ$ with a representative element in~$M F^\vee$, i.e.
\begin{equation*}
 \Lambda_M\equiv M F^\vee \cap P/MQ.
 \end{equation*}
Considering a natural action of $W$ on the quotient group $\R^n/MQ$, we denote for $\la \in \R^n/MQ$ the isotropy group and its order by
\begin{equation}\label{hla}
\mathrm{Stab}^{\vee} (\la)=\set{w\in W}{w\la=\la},\q h^{\vee}_{\la}\equiv |\mathrm{Stab}^{\vee} (\la)|.
\end{equation}
Recall the following three properties from Proposition 3.6 in \cite{HP} of the action of $W$ on the quotient group $\R^n/MQ$.
\begin{enumerate}
\item For any $\la\in P/MQ$, there exists $\la'\in\Lambda_M  $ and $w\in W$ such that
\begin{equation}\label{dfun1}
 \la=w\la'.
\end{equation}
\item If $\la,\la'\in \Lambda_M $ and $\la'=w\la$, $w\in W$, then
\begin{equation}\label{lrfun2}
 \la'=\la=w\la.
\end{equation}
\item If $\la\in M F^\vee \cap \R^n/MQ $, i.e. $\la=b+MQ$, $b\in MF^\vee$, then $\widehat\psi (\mathrm{Stab}_{\widehat{W}^{\mathrm{aff}}}(b/M))=  \mathrm{Stab}^{\vee} (\la)$ and
\begin{equation}\label{rfunstab2}
\mathrm{Stab}^{\vee} (\la) \cong \mathrm{Stab}_{\widehat{W}^{\mathrm{aff}}}(b/M).
\end{equation}
\end{enumerate}

\section{Sign homomorphisms and orbit functions}

\subsection{Sign homomorphisms}\

The Weyl group $W$ has the following abstract presentation \cite{BB,H2}
\begin{equation}\label{presentation}r_i^2=1,\quad (r_ir_j)^{m_{ij}}=1,\quad i,j=1,\dots,n\end{equation}
where integers $m_{ij}$ are elements of the Coxeter matrix.
To introduce various classes of orbit functions, we consider 'sign' homomorphisms $\sigma:W\rightarrow\{\pm1 \}.$ An admissible mapping $\sigma$ must satisfy the presentation condition \eqref{presentation}
\begin{equation}\label{admit}
\sigma(r_i)^2=1,\quad (\sigma(r_i)\sigma(r_j))^{m_{ij}}=1,\quad i,j=1,\dots,n.
\end{equation}

If condition \eqref{admit} is satisfied, then it follows from the universality property (see e.g. \cite{BB}) that $\sigma$ is a well-defined homomorphism and its values on any $w\in W$ are given as products of generator values. The following two choices of homomorphism values of generators $r_\al,\,\al\in\Delta$, obviously satisfying \eqref{admit}, lead to the well-known homomorphisms:
\begin{align}
\id(r_\al)&=1 \label{ghomid} \\
\sigma^e(r_\al)&=-1 \label{ghome}
\end{align}
which yield for any $w\in W$
\begin{align}
\id(w)&=1  \\
\sigma^e(w)&=\det w. \label{parity}
\end{align}

It is shown in \cite{MMP} that, for root systems with two different lengths of roots, there are two other available choices. Using the decomposition \eqref{sl}, these two new homomorphisms are given as follows \cite{MMP}:
\begin{align}
\sigma^s(r_\al)&=\begin{cases} 1,\quad \al\in \Delta_l  \\ -1,\quad \al\in \Delta_s\end{cases}\\
\sigma^l(r_\al)&=\begin{cases} 1,\quad \al\in \Delta_s  \\ -1,\quad \al\in \Delta_l.\end{cases}
\end{align}

Since the highest root $\xi\in W \Delta_l$, there exist $w\in W$ and $\al\in\Delta_l $ such that $\xi=w\al$. Then from the relation $r_\xi=wr_\al w^{-1}$, we obtain for any sign homomorphism that $\sigma (r_\xi)=\sigma (r_\al)$ holds. Thus we have
\begin{equation}\label{slxi}
\sigma^s(r_\xi)=1,\q \sigma^l(r_\xi)=-1.
\end{equation}
Similarly, for the highest dual root $\eta$ there exists a root $\beta\in W \Delta_s$ such that $\eta=2\beta /\sca{\beta}{\beta}$, and we obtain
\begin{equation}\label{sleta}
\sigma^s(r_\eta)=-1,\q \sigma^l(r_\eta)=1.
\end{equation}

\subsection{Fundamental domains}\

Each of the sign homomorphisms $\sigma^s$ and $\sigma^l$ determines a decomposition of the fundamental domain~$F$. The factors of this decomposition will be crucial for the study of the orbit functions. We introduce two subsets of $F$:
\begin{align*}
F^s&=\setm{a\in F}{\sigma^s \circ \psi \left(\mathrm{Stab}_{W^{\mathrm{aff}}}(a)\right)=\{1\} }\\
F^l&=\setm{a\in F}{\sigma^l \circ \psi \left(\mathrm{Stab}_{W^{\mathrm{aff}}}(a)\right)=\{1\} }
\end{align*}
where $\psi$ is the retraction homomorphism \eqref{ret}. Since for all points of the interior of $F$ the stabilizer is trivial, i.e. $\mathrm{Stab}_{W^{\mathrm{aff}}}(a)=1$, $a\in \mathrm{int} (F)$, the interior $\mathrm{int} (F)$ is a subset of both $F^s$ and $F^l$. In order to determine the analytic form of the sets $F^s$ and $F^l$, we define two subsets of the boundaries of $F$
\begin{align*}
H^s&=\set{a\in F}{(\exists r\in R^s)(ra=a)}\\
H^l&=\set{a\in F}{(\exists r\in R^l)(ra=a)}.
\end{align*}
Note that, since for the affine reflection $r_0\in R^l$ it holds that $\psi(r_0)=r_\xi$, we have from \eqref{slxi} that $\sigma^s \circ \psi (r_0)=1$ and $\sigma^l \circ \psi (r_0)=-1$. Taking into account the disjoint decomposition \eqref{slaff}, we obtain for any $r\in R$ the following two exclusive choices:
\begin{equation}\label{rs}
\begin{alignedat}{4}
\sigma^s \circ \psi (r)&=-1, &\q \sigma^l \circ \psi (r)&=1,  &\q r&\in R^s\\
\sigma^s \circ \psi (r)&=1, &\q \sigma^l \circ \psi (r)&=-1,  &\q r&\in R^l.
\end{alignedat}
\end{equation}

\begin{tvr}\label{FsFl}For the sets $F^s$ and $F^l$, the following holds:
\begin{enumerate}
\item $F^s=F\setminus H^s$.
\item $F^l=F\setminus H^l$.
\end{enumerate}
\end{tvr}
\begin{proof}
Let $a\in F$.
\begin{enumerate}
\item If $a\notin F\setminus H^s$, then $a\in H^s$, and there exists $r\in R^s$ such that $r\in \mathrm{Stab}_{W^{\mathrm{aff}}}(a)$. Then according to \eqref{rs}, we have $\sigma^s \circ \psi (r)=-1$. Thus, $\sigma^s \circ \psi \left(\mathrm{Stab}_{W^{\mathrm{aff}}}(a)\right)=\{\pm 1\}$ and consequently $a\notin F^s$. Conversely, if $a\in F\setminus H^s$, then the stabilizer $\mathrm{Stab}_{W^{\mathrm{aff}}}(a)$ is either trivial or generated by generators from $R^l$ only. Then, since for any generator $r\in R^l$ it follows from \eqref{rs} that $\sigma^s \circ \psi (r)=1$, we obtain $\sigma^s \circ \psi \left(\mathrm{Stab}_{W^{\mathrm{aff}}}(a)\right)=\{ 1\}$, i.e. $a\in F^s$.
\item This case is completely analogous to case (1).
\end{enumerate}
\end{proof}
The explicit description of domains $F^s$ and $F^l$ now follows from \eqref{deffun} and Proposition \ref{FsFl}. We introduce the symbols $y^s_i$, $y^l_i\in\R$, $i=0,\dots,n$ in the following way:
\begin{equation}\label{ysyl}
\begin{alignedat}{4}
y^s_i&> 0, &\q y^l_i&\geq 0,  &\q r_i&\in R^s\\
y^s_i&\geq 0, &\q y^l_i&> 0,  &\q r_i&\in R^l.
\end{alignedat}
\end{equation}
Thus, the explicit form of $F^s$ and $F^l$ is given by
\begin{equation}\label{FsFlex}
\begin{alignedat}{2}
F^s &=\setb{y^s_1\om^{\vee}_1+\dots+y^s_n\om^{\vee}_n}{y^s_0+y^s_1 m_1+\dots+y^s_n m_n=1  }  \\
F^l &=\setb{y^l_1\om^{\vee}_1+\dots+y^l_n\om^{\vee}_n}{y^l_0+y^l_1 m_1+\dots+y^l_n m_n=1  } .
\end{alignedat}
\end{equation}

\subsection{Dual fundamental domains}\

The sign homomorphisms $\sigma^s$ and $\sigma^l$ also determine a decomposition of the dual fundamental domain~$F^\vee$. The factors of this decomposition will be needed for the study of the discretized orbit functions. We introduce two subsets of $F^\vee$:
\begin{equation}\label{FsFldual}
\begin{alignedat}{2}
F^{s\vee}&=\setm{a\in F^\vee}{\sigma^s \circ \widehat\psi \left(\mathrm{Stab}_{\widehat{W}^{\mathrm{aff}}}(a)\right)=\{1\} }\\
F^{l\vee}&=\setm{a\in F^\vee}{\sigma^l \circ \widehat\psi \left(\mathrm{Stab}_{\widehat{W}^{\mathrm{aff}}}(a)\right)=\{1\} }
\end{alignedat}
\end{equation}
where $\widehat\psi$ is the dual retraction homomorphism \eqref{retd}. Since for all points of the interior of $F^\vee$ the stabilizer is trivial, i.e. $\mathrm{Stab}_{\widehat{W}^{\mathrm{aff}}}(a)=1$, $a\in \mathrm{int} (F^\vee)$, the interior $\mathrm{int} (F^\vee)$ is a subset of both $F^{s\vee}$ and $F^{l\vee}$. In order to determine the analytic form of  sets $F^{s\vee}$ and $F^{l\vee}$, we define two subsets of the boundaries of $F^\vee$
\begin{align*}
H^{s\vee}&=\set{a\in F^\vee}{(\exists r\in R^{s\vee})(ra=a)}\\
H^{l\vee}&=\set{a\in F^\vee}{(\exists r\in R^{l\vee})(ra=a)}.
\end{align*}
Note that, since for the affine reflection $r^\vee_0\in R^{s\vee}$ it holds that $\widehat\psi(r^\vee_0)=r_\eta$, we have from \eqref{sleta} that $\sigma^s \circ \widehat\psi (r^\vee_0)=-1$ and $\sigma^l \circ \widehat\psi (r^\vee_0)=1$. Taking into account the disjoint decomposition \eqref{slaffdual}, we obtain for any $r\in R^\vee$ the following two exclusive choices:
\begin{equation}\label{rsd}
\begin{alignedat}{4}
\sigma^s \circ \widehat\psi (r)&=-1, &\q \sigma^l \circ \widehat\psi (r)&=1,  &\q r&\in R^{s\vee}\\
\sigma^s \circ \widehat\psi (r)&=1, &\q \sigma^l \circ \widehat\psi (r)&=-1,  &\q r&\in R^{l\vee}.
\end{alignedat}
\end{equation}
Similarly to Proposition \ref{FsFl}, we obtain:
\begin{tvr}\label{FsFld}For sets $F^{s\vee}$ and $F^{l\vee}$, the following holds:
\begin{enumerate}
\item $F^{s\vee}=F^\vee\setminus H^{s\vee}$.
\item $F^{l\vee}=F^\vee\setminus H^{l\vee}$.
\end{enumerate}
\end{tvr}

The explicit description of domains $F^{s\vee}$ and $F^{l\vee}$ now follows from \eqref{deffund} and Proposition \ref{FsFld}. We introduce the symbols $z^s_i$, $z^l_i\in\R$, $i=0,\dots,n$ in the following way:
\begin{equation}\label{zszl}
\begin{alignedat}{4}
z^s_i&> 0, &\q z^l_i&\geq 0,  &\q r_i&\in R^{s\vee}\\
z^s_i&\geq 0, &\q z^l_i&> 0,  &\q r_i&\in R^{l\vee}.
\end{alignedat}
\end{equation}
Thus, the explicit form of $F^{s\vee}$ and $F^{l\vee}$ is given by
\begin{equation}\label{FsFlexdual}
\begin{alignedat}{2}
F^{s\vee} &=\setb{z^s_1\om_1+\dots+z^s_n\om_n}{z^s_0+z^s_1 m^\vee_1+\dots+z^s_n m^\vee_n=1  }  \\
F^{l\vee} &=\setb{z^l_1\om_1+\dots+z^l_n\om_n}{z^l_0+z^l_1 m^\vee_1+\dots+z^l_n m^\vee_n=1  } .
\end{alignedat}
\end{equation}

\section{$S^s-$ and $S^l-$functions}\

Four sign homomorphisms $\id$, $\sigma^e$, $\sigma^s$ and $\sigma^l$ induce four types of families of complex orbit functions. Within each family, determined by $\sigma\in \{\id,\,\sigma^e,\, \sigma^s,\,\sigma^l \}$, are the complex functions $\phi^\sigma_b:\R^n\map \C$ labeled by weights $b\in P$ and in the general form
\begin{equation}\label{genorb}
\phi^\sigma_b(a)=\sum_{w\in W}\sigma (w)\, e^{2 \pi i \sca{ wb}{a}},\q a\in \R^n.
\end{equation}
The resulting functions for $\sigma=\id$ in (\ref{genorb}) are called $C-$functions; for their detailed review, see~\cite{KP2}. For $\sigma=\sigma^e$,
we obtain the well-known $S-$functions \cite{KP3}. The discretization properties of both $C-$ and $S-$functions on a finite fragment
of the grid $\frac{1}{M}P^{\vee}$ were described in \cite{HP}. The remaining two options of homomorphisms $\sigma^s$ and $\sigma^l$ and
corresponding functions $\phi^{\sigma^s}_\la$, $\phi^{\sigma^l}_\la$, called $S^l-$ and $S^s-$functions \cite{MMP}, were studied in detail for
$G_2$ only \cite{Sz}. In order to describe the discretization of functions $\phi^{\sigma^s}_\la$ and $\phi^{\sigma^l}_\la$ in full generality, we first review their basic properties.

\subsection{$S^s-$functions}\

\subsubsection{Symmetries of $S^s-$functions}
Choosing $\sigma=\sigma^s$ in \eqref{genorb}, we obtain $S^s-$functions $\phi^{\sigma^s}_b$; we abbreviate the notation by denoting $\phi^s_b\equiv \phi^{\sigma^s}_b$, i.e.
\begin{equation}\label{genorbs}
\phi^s_b(a)=\sum_{w\in W}\sigma^s (w)\, e^{2 \pi i \sca{ wb}{a}},\q a\in \R^n, \, b\in P.
\end{equation}
The following properties of $S^s-$functions are crucial:
\begin{itemize}
\item (anti)symmetry with respect to $w\in W$
\begin{align}\label{Sssym}
\phi^s_{b}(wa)&=\sigma^s(w)\phi^s_{b}(a)\\
\phi^s_{wb}(a)&=\sigma^s(w)\phi^s_{b}(a)\label{Sssymla}
 \end{align}
\item invariance with respect to shifts from $q^{\vee} \in Q^{\vee}$
\begin{equation}\label{Ssinv}
\phi^s_{b}(a+q^\vee)=\phi^s_{b}(a).
 \end{equation}
\end{itemize}
Thus, the $S^s-$functions are (anti)symmetric with respect to the affine Weyl group $W^{\mathrm{aff}}$. This allows us to consider the values $\phi^s_{b}(a)$ only for points of the fundamental domain $a\in F$. Moreover,
from \eqref{rs}, \eqref{Sssym} we deduce that
\begin{equation}\label{bound2}
\phi^s_{b}(ra)=-\phi^s_{b}(a),\q  r\in R^s.
\end{equation}
This antisymmetry implies that the functions $\phi^s_{b}$ are for all $b\in P$ zero on part $H^s$ of the boundary of $F$
\begin{equation}\label{Fss}
\phi^s_{b}(a')=0,\q  a'\in H^s
\end{equation}
and therefore {\it we consider the functions $\phi^s_{b}$ on the fundamental domain $F^s=F\setminus H^s$ only}.

\subsubsection{Discretization of $S^s-$functions}

In order to develop discrete calculus of $S^s-$functions, we investigate the behavior of these functions on the grid $\frac{1}{M}P^{\vee}$. Suppose we have fixed $M\in \N$ and $u\in \frac{1}{M}P^{\vee}$. It follows from (\ref{Ssinv}) that we can consider $\phi^s_b$ as a function on cosets from $\frac{1}{M}P^{\vee}/Q^{\vee}$. It follows from (\ref{Fss}) that we can consider $\phi^s_b$ only on the set
\begin{equation}
F^s_M\equiv\frac{1}{M}P^{\vee}/Q^{\vee}\cap F^s.
\end{equation}
Next we have
\begin{equation*}
 \phi^s_{b+MQ}(u)=\phi^s_b(u), \q u\in F_M^s
 \end{equation*}
and thus we can consider the functions $\phi^s_\la$ on $F_M^s$ parametrized by cosets from $\la\in P/MQ$. Moreover, it follows from (\ref{dfun1}) and (\ref{Sssymla}) that we can consider $\phi^s_\la$ on $F_M^s$ parameterized by classes from $\Lambda_M$. Taking any $\la\in \Lambda_M$ and any reflection $r^\vee\in R^{s\vee}$ we calculate directly using \eqref{sleta}, \eqref{Sssymla} and \eqref{Ssinv} that
\begin{equation*}
\phi^s_{Mr^\vee (\frac{\la}{M})}( u)=- \phi^s_\la(u),\q u\in F_M^s.
\end{equation*}
This implies that, for $\la \in M H^{s\vee}\cap \Lambda_M $, the functions $\phi^s_\la$ are zero on $F_M^s$, i.e.
$$
 \phi^s_{\la}( u)=0,\q \la \in M H^{s\vee}\cap \Lambda_M,\, u\in F_M^s.
$$
Defining the set
\begin{equation}\label{Ls}
\Lambda^s_M\equiv P/MQ \cap MF^{s\vee}
\end{equation}
we conclude that {\it we can consider $S^s-$functions $\phi^s_\la$ on $F_M^s$ parameterized by $\la\in\Lambda_M^s$ only}.

\subsection{$S^l-$functions}\

\subsubsection{Symmetries of $S^l-$functions}
Choosing $\sigma=\sigma^l$ in \eqref{genorb}, we obtain $S^l-$functions $\phi^{\sigma^l}_b$; we abbreviate the notation by denoting $\phi^l_b\equiv \phi^{\sigma^l}_b$, i.e.
\begin{equation}\label{genorbl}
\phi^l_b(a)=\sum_{w\in W}\sigma^l (w)\, e^{2 \pi i \sca{ wb}{a}},\q a\in \R^n, \, b\in P.
\end{equation}
The following properties of $S^l-$functions are crucial
\begin{itemize}
\item (anti)symmetry with respect to $w\in W$
\begin{align}\label{Slsym}
\phi^l_{b}(wa)&=\sigma^l(w)\phi^l_{b}(a)\\
\phi^l_{wb}(a)&=\sigma^l(w)\phi^l_{b}(a)\label{Slsymla}
 \end{align}
\item invariance with respect to shifts from $q^{\vee} \in Q^{\vee}$
\begin{equation}\label{Slinv}
\phi^l_{b}(a+q^\vee)=\phi^l_{b}(a).
 \end{equation}
\end{itemize}
Thus, the $S^l-$functions are (anti)symmetric with respect to the affine Weyl group $W^{\mathrm{aff}}$. This allows us to consider the values $\phi^l_{b}(a)$ only for points of the fundamental domain $a\in F$. Moreover,
from \eqref{rs}, \eqref{Slsym} and \eqref{Slinv} we deduce that
\begin{equation}\label{bound2l}
\phi^l_{b}(ra)=-\phi^l_{b}(a),\q  r\in R^l.
\end{equation}
This antisymmetry implies that the functions $\phi^l_{b}$ are for all $b\in P$ zero on part $H^l$ of the boundary of $F$
\begin{equation}\label{Fls}
\phi^l_{b}(a')=0,\q  a'\in H^l
\end{equation}
and therefore {\it we consider the functions $\phi^l_{b}$ on the fundamental domain $F^l=F\setminus H^l$ only}.

\subsubsection{Discretization of $S^l-$functions}

In order to develop discrete calculus of $S^l-$functions, we investigate the behavior of these functions on the grid $\frac{1}{M}P^{\vee}$. Suppose we have fixed $M\in \N$ and $u\in \frac{1}{M}P^{\vee}$. It follows from (\ref{Slinv}) that we can consider $\phi^l_b$ as a function on cosets from $\frac{1}{M}P^{\vee}/Q^{\vee}$. It follows from (\ref{Fls}) that we can consider $\phi^l_b$ only on the set
\begin{equation}
F^l_M\equiv\frac{1}{M}P^{\vee}/Q^{\vee}\cap F^l.
\end{equation}
Next we have
\begin{equation*}
 \phi^l_{b+MQ}(u)=\phi^l_b(u), \q u\in F_M^l
 \end{equation*}
and thus we can consider the functions $\phi^l_\la$ on $F_M^l$ parametrized by cosets from $\la\in P/MQ$. Moreover, it follows from (\ref{dfun1}) and (\ref{Slsymla}) that we can consider $\phi^l_\la$ on $F_M^l$ parameterized by classes from $\Lambda_M$. Taking any $\la\in \Lambda_M$ and any reflection $r^\vee\in R^{l\vee}$, we calculate directly using \eqref{sleta}, \eqref{Slsymla} that
\begin{equation*}
\phi^l_{r^\vee \la}( u)=- \phi^l_\la(u),\q u\in F_M^l.
\end{equation*}
This implies that for $\la \in M H^{l\vee}\cap \Lambda_M $, the functions $\phi^l_\la$ are zero on $F_M^l$, i.e.
$$
 \phi^l_{\la}( u)=0,\q \la \in M H^{l\vee}\cap \Lambda_M,\, u\in F_M^l.
$$
Defining the set
\begin{equation}\label{Ll}
\Lambda^l_M\equiv P/MQ \cap MF^{l\vee}
\end{equation}
we conclude that {\it we can consider $S^l-$functions $\phi^l_\la$ on $F_M^l$ parameterized by $\la\in\Lambda_M^l$ only}.

\section{Number of grid elements}

\subsection{Number of elements of $F_M^s$ and $F_M^l$}\

Recall from \cite{HP} that, for an arbitrary $M\in\N$, the grid $F_M$ is given as cosets from the $W-$invariant group $\frac{1}{M}P^{\vee}/Q^{\vee}$ with a representative element in the fundamental domain $F$
\begin{equation*}
F_M\equiv\frac{1}{M}P^{\vee}/Q^{\vee}\cap F
 \end{equation*}
and the following property holds
\begin{equation}\label{WFM}
 WF_M = \frac{1}{M}P^{\vee}/Q^{\vee}.
\end{equation}
The representative points of $F_M$ can be explicitly written as
\begin{equation}\label{FM}
 F_M = \setb{\frac{u_1}{M}\om^{\vee}_1+\dots+\frac{u_n}{M}\om^{\vee}_n}{u_0,u_1,\dots ,u_n \in \Z^{\geq 0},\, u_0+u_1m_1+\dots + u_n m_n=M}.
 \end{equation}
The number of elements of $F_M$, denoted by $|F_M|$, are also calculated in \cite{HP} for all simple Lie algebras. Using these results, we derive the number of elements of $F_M^s$ and $F_M^l$. Firstly, we describe explicitly the sets $F_M^s$ and $F_M^l$. Similarly to \eqref{ysyl}, we introduce the symbols $u^s_i$, $u^l_i\in\R$, $i=0,\dots,n$:
\begin{equation}
\begin{alignedat}{4}
u^s_i&\in \N, &\q u^l_i&\in \Z^{\geq 0},  &\q r_i&\in R^s\\
u^s_i&\in \Z^{\geq 0}, &\q u^l_i&\in\N,  &\q r_i&\in R^l.
\end{alignedat}
\end{equation}
The explicit form of $F^s_M$ and $F^l_M$ then follows from the explicit form of $F^s$ and $F^l$ in \eqref{FsFlex}:
\begin{align}
F^s_M &=\setb{\frac{u^s_1}{M}\om^{\vee}_1+\dots+\frac{u^s_n}{M}\om^{\vee}_n}{u^s_0+u^s_1 m_1+\dots+u^s_n m_n=M  } \label{Fsex} \\
F^l_M &=\setb{\frac{u^l_1}{M}\om^{\vee}_1+\dots+\frac{u^l_n}{M}\om^{\vee}_n}{u^l_0+u^l_1 m_1+\dots+u^l_n m_n=M  } .
\end{align}
Using the following proposition, the number of elements of $F_M^s$ and $F_M^l$  can be obtained from the formulas for $|F_M|$.
\begin{tvr}\label{Cox}
Let $m^s$ and $m^l$ be the short and long Coxeter numbers, respectively. Then
\begin{equation}\label{numFsFl}
| F^s_M|=\begin{cases}0 & M<m^s \\ 1 & M=m^s \\ |F_{M-m^s}| & M>m^s .\end{cases},\q | F^l_M|=\begin{cases}0 & M<m^l \\ 1 & M=m^l \\ |F_{M-m^l}| & M>m^l .\end{cases}
\end{equation}
\end{tvr}
\begin{proof}
Taking non-negative numbers $u_i\in \Z^{\geq 0}$ and substituting the relations $ u^s_i=1+u_i$ if $r_i\in R^s$ and $ u^s_i=u_i$ if $r_i\in R^l$ into the defining relation (\ref{Fsex}), we obtain
\begin{equation*}
 u_0+m_1u_1+\dots + m_n u_n= M-m^s,\q u_0,\dots,u_n\in \Z^{\geq 0}.
\end{equation*}
This equation has one solution $[0,\dots,0]$ if $M=m^s$, no solution if $M<m^s$, and is equal to the defining relation (\ref{FM}) of $F_{M-m^s}$ if $M>m^s$. The case of $F_M^l$ is similar.
\end{proof}


\begin{thm}\label{numAn}
The numbers of points of grids $F_M^s$ and $F_M^l$ of Lie algebras $B_n$, $C_n$, $G_2$ and $F_4$ are given by the following relations.
\begin{enumerate}\item  $C_n,\,n\geq 2$,
\begin{align*}
&|F^s_{2k}(C_n)|= \begin{pmatrix}k+1 \\ n \end{pmatrix}+\begin{pmatrix}k \\ n \end{pmatrix} \\
&|F^s_{2k+1}(C_n)|= 2\begin{pmatrix}k+1 \\ n \end{pmatrix}
\end{align*}
\begin{align*}
&|F^l_{2k}(C_n)|= \begin{pmatrix}n+k-1 \\ n \end{pmatrix}+\begin{pmatrix}n+k-2 \\ n \end{pmatrix}\\
&|F^l_{2k+1}(C_n)|= 2\begin{pmatrix}n+k-1 \\ n \end{pmatrix}
\end{align*}
\item $B_n,\,n\geq 3$,
$$|F_M^s(B_n)|=|F_M^l(C_n)|$$
$$|F_M^l(B_n)|=|F_M^s(C_n)|$$
\item $G_2$
\begin{equation*}
\begin{alignedat}{4}
|F^s_{6k}(G_2)|&= 3k^2, &\qquad |F^s_{6k+1}(G_2)|&= 3k^2+k\\
|F^s_{6k+2}(G_2)|&= 3k^2+2k, &\qquad |F^s_{6k+3}(G_2)|&=3k^2+3k+1\\
|F^s_{6k+4}(G_2)|&=3k^2+4k+1,  &\qquad |F^s_{6k+5}(G_2)|&=3k^2+5k+2.
\end{alignedat}
\end{equation*}
$$|F^l_{M}(G_2)|=|F^s_M(G_2)|$$
\item $F_4$
\begin{equation*}
\begin{alignedat}{4}
|F^s_{12k}(F_4)|&= 18k^4-k^2, &\qquad |F^s_{12k+1}(F_4)|&= 18k^4+6k^3-\frac52k^2-\frac12k\\
|F^s_{12k+2}(F_4)|&= 18k^4+12k^3+2k^2, &\qquad |F^s_{12k+3}(F_4)|&=18k^4+18k^3+\frac72k^2-\frac12k\\
|F^s_{12k+4}(F_4)|&=18k^4+24k^3+11k^2+2k,  &\qquad |F^s_{12k+5}(F_4)|&=18k^4+30k^3+\frac{31}{2}k^2+\frac52k\\
|F^s_{12k+6}(F_4)|&=18k^4+36k^3+26k^2+8k+1,  &\qquad |F^s_{12k+7}(F_4)|&=18k^4+42k^3+\frac{67}{2}k^2+\frac{21}{2}k+1\\
|F^s_{12k+8}(F_4)|&=18k^4+48k^3+47k^2+20k+3,  &\qquad |F^s_{12k+9}(F_4)|&=18k^4+54k^3+\frac{115}{2}k^2+\frac{51}{2}k+4\\
|F^s_{12k+10}(F_4)|&=18k^4+60k^3+74k^2+40k+8,  &\qquad |F^s_{12k+11}(F_4)|&=18k^4+66k^3+\frac{175}{2}k^2+\frac{99}{2}k+10\\
\end{alignedat}
\end{equation*}
$$|F^l_M(F_4)|=|F^s_M(F_4)|.$$
\end{enumerate}
\end{thm}
\begin{proof}
For the case $C_n$, we have that $|F_{2k}(C_n)|=\begin{pmatrix}n+k \\ n \end{pmatrix}+\begin{pmatrix}n+k-1 \\ n \end{pmatrix}$ from \cite{HP} and  $m^s=2n-2$ from Table \ref{tabdec}. It can be verified directly that the formula
 $$|F^s_{2k}(C_n)|= \begin{pmatrix}k+1 \\ n \end{pmatrix}+\begin{pmatrix}k \\ n \end{pmatrix}$$
satisfies (\ref{numFsFl}) for all values of $k\in \N$. Analogously, we obtain formulas for the remaining cases.
\end{proof}
\begin{example}\label{ex1}
For the Lie algebra $C_2$, we have Coxeter number $m=4$ and $c=2$. For $M=4$, the order of the group $\frac{1}{4}P^{\vee}/Q^{\vee}$ is equal to $32$, and according to Theorem \ref{numAn} we calculate $$\abs{F^s_4(C_2)}=|F^l_4(C_2)|=\comb{3}{2}+\comb{2}{2}=4. $$
The coset representatives of $\frac{1}{4}P^{\vee}/Q^{\vee}$ and the fundamental domains $F^s$ and $F^l$ are depicted in Figure~\ref{figC2}.
\begin{figure}
\resizebox{11cm}{!}{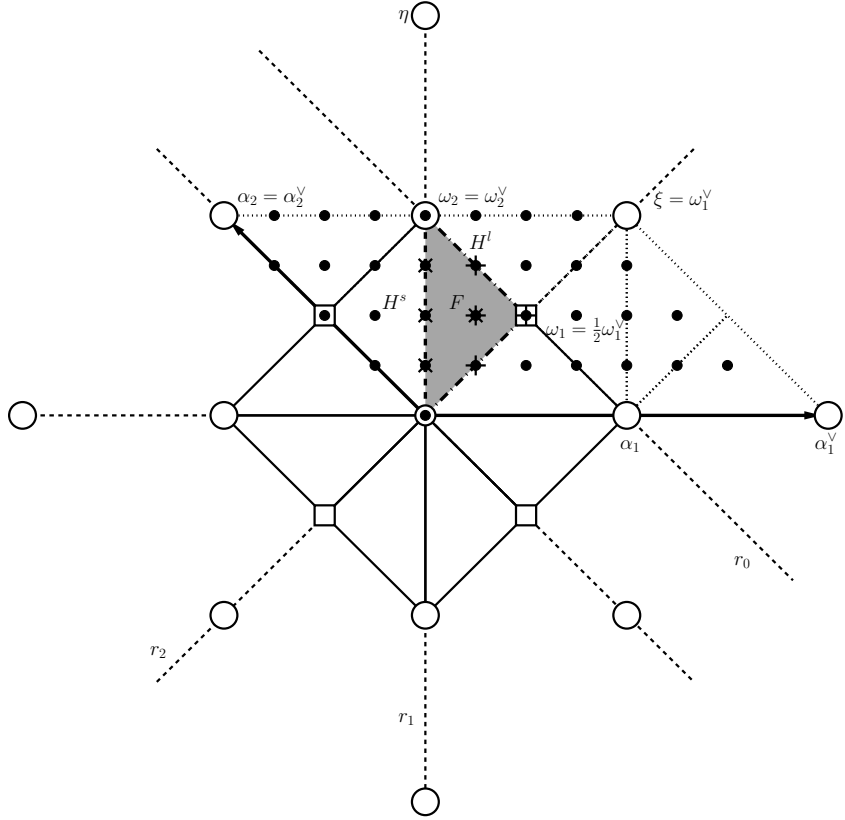}
\caption{ The fundamental domains $F^s$ and $F^l$ of $C_2$. The fundamental domain $F$ is depicted as the grey triangle containing borders $H^s$ and $H^l$, depicted as the thick dashed line and dot-and-dashed lines, respectively. The coset representatives of $\frac{1}{4}P^{\vee}/Q^{\vee}$ are shown as $32$ black dots. The four representatives belonging to $F^s_4$ and $F^l_4$ are crossed with '$+$' and '$\times$', respectively. The dashed lines represent 'mirrors' $r_0,r_1$ and $r_2$. Circles are elements of the root lattice $Q$; together with the squares they are elements of the weight lattice $P$.  }\label{figC2}
\end{figure}
\end{example}
\subsection{Number of elements of $\Lambda_M^s$ and $\Lambda_M^l$}\

In this section, we relate the numbers of elements of $F_M^s$, $F_M^l$ to the numbers of elements $\Lambda_M^s$, $\Lambda_M^l$, defined by \eqref{Ls}, \eqref{Ll}.
Firstly, we describe explicitly the sets $\Lambda_M^s$ and $\Lambda_M^l$.
Similarly to \eqref{zszl}, we introduce the symbols $t^s_i$, $t^l_i\in\R$, $i=0,\dots,n$:
\begin{equation}
\begin{alignedat}{4}
t^s_i&\in \N, &\q t^l_i&\in \Z^{\geq 0},  &\q r_i&\in R^{s\vee}\\
t^s_i&\in \Z^{\geq 0}, &\q t^l_i&\in\N,  &\q r_i&\in R^{l\vee}.
\end{alignedat}
\end{equation}
The explicit form of $\Lambda^s_M$ and $\Lambda^l_M$ then follows from the explicit form of $F^{s\vee}$ and $F^{l\vee}$ in \eqref{FsFlexdual}:
\begin{equation}
\begin{alignedat}{2}
\Lambda^s_M &=\setb{t^s_1\om_1+\dots+t^s_n\om_n}{t^s_0+t^s_1 m^\vee_1+\dots+t^s_n m^\vee_n=M  } \\
\Lambda^l_M &=\setb{t^l_1\om_1+\dots+t^l_n\om_n}{t^l_0+t^l_1 m^\vee_1+\dots+t^l_n m^\vee_n=M  } .
\end{alignedat}
\end{equation}
Similarly to Proposition \ref{Cox}, we obtain the following one.
\begin{tvr}\label{Cox2}
Let $m^{s\vee}$ and $m^{l\vee}$ be the short and the long dual Coxeter numbers, respectively. Then
\begin{equation}\label{numLsLl}
| \Lambda^s_M|=\begin{cases}0 & M<m^{s\vee} \\ 1 & M=m^{s\vee} \\ |F_{M-m^{s\vee}}| & M>m^{s\vee} \end{cases},\q | \Lambda^l_M|=\begin{cases}0 & M<m^{l\vee} \\ 1 & M=m^{l\vee} \\ |F_{M-m^{l\vee}}| & M>m^{l\vee} .\end{cases}
\end{equation}
\end{tvr}
Combining Propositions \ref{Cox}, \ref{Cox2} and \ref{msml} and taking into account that $|F_M|=|\Lambda_M|$ we conclude with the following crucial result.
\begin{cor}\label{FL}
For the numbers of elements of the sets $\Lambda^s_M$ and $\Lambda^l_M$ it holds that
\begin{equation}\label{complete}
\begin{alignedat}{2}
|\Lambda^s_M |&=|F^s_M|, \\
|\Lambda^l_M |&=|F^l_M| .
\end{alignedat}
\end{equation}
\end{cor}
\begin{example}\label{exdual}
For the Lie algebra $C_2$ we have  $\abs{P/4Q}=32$ and according to Theorem \ref{numAn} and Corollary~\ref{FL} we have $$|\Lambda^s_4(C_2)|=|\Lambda^l_4(C_2)|=4. $$
The cosets representants of $P/4Q$ together with the grids of weights $\Lambda^s_4(C_2)$ and $\Lambda^l_4(C_2)$ are depicted in Figure~\ref{figC2d}.
\begin{figure}
\resizebox{11cm}{!}{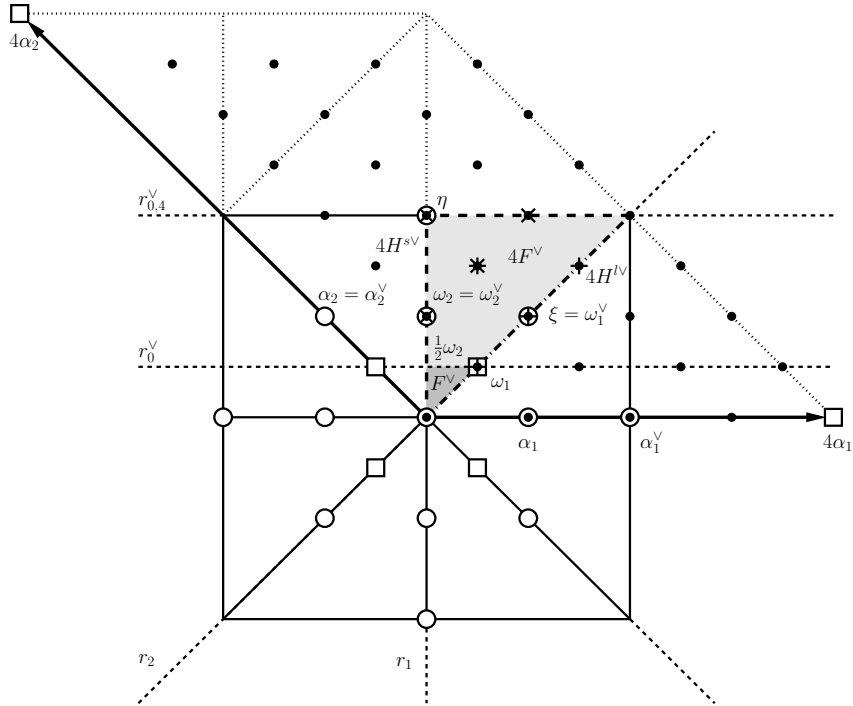}
\caption{ The grids of weights $\Lambda^s_4(C_2)$ and $\Lambda^l_4(C_2)$ of $C_2$. The darker grey triangle is the fundamental domain $F^{\vee}$ and the lighter grey triangle is the domain $4F^{\vee}$. The borders $4H^{s\vee}$ and $4H^{l\vee}$ are depicted as the thick dashed lines and dot-and-dashed lines, respectively. The cosets representants of $P/4Q$ of $C_2$ are shown as $32$ black dots. The four representants belonging to $\Lambda^s_4(C_2)$ and $\Lambda^l_4(C_2)$ are crossed with '$+$' and '$\times$', respectively. The dashed lines represent dual 'mirrors' $r^\vee_0,r_1$, $r_2$ and the affine mirror $r^\vee_{0,4}$ is defined by $r^\vee_{0,4}\la= 4r^\vee_0(\la/4)$. The circles and squares coincide with those in Figure \ref{figC2}. }\label{figC2d}
\end{figure}
\end{example}

\section{Discrete orthogonality and transforms of $S^s-$ and $S^l-$ functions}

\subsection{Discrete orthogonality of $S^s-$ and $S^l-$functions}\

To describe the discrete orthogonality of the $S^l-$ and $S^s-$functions, we use the ideas discussed in \cite{MP2} and reformulated in \cite{HP}. Recall that basic orthogonality relations from \cite{HP,MP2} are, for any $\la,\la' \in P/MQ$, of the form:
\begin{equation}\label{bdis}
 \sum_{y\in\frac{1}{M}P^{\vee}/Q^{\vee}} e^{2\pi i\sca{\la-\la'}{y}}=cM^n\delta_{\la,\la'}.
\end{equation}
We define the scalar product of two functions $f,g:F_M^s\map \Com$ or $f,g:F_M^l\map \Com$ by
\begin{equation} \label{scp}
\sca{f}{g}_{F_M^s}= \sum_{x\in F_M^s}\ep(x) f(x)\overline{g(x)},\q \sca{f}{g}_{F_M^l}= \sum_{x\in F_M^l}\ep(x) f(x)\overline{g(x)},
\end{equation}
where the numbers $\ep (x)$ are determined by (\ref{ep}). We show that $\Lambda_M^s$ and $\Lambda_M^l$ are the lowest maximal sets of pairwise orthogonal $S^s-$ and $S^l-$functions.
\begin{thm}
For $\la,\la' \in\Lambda_M^s$ it holds that
\begin{equation}\label{ortho}
 \sca{\phi^s_\la}{\phi^s_{\la'}}_{F_M^s}=c\abs{W}M^n h^{\vee}_\la \delta_{\la,\la'}
\end{equation}
and for $\la,\la' \in\Lambda_M^l$ it holds that
\begin{equation}\label{orthol}
 \sca{\phi^l_\la}{\phi^l_{\la'}}_{F_M^l}=c\abs{W}M^n h^{\vee}_\la \delta_{\la,\la'},
\end{equation}
where $c$, $h^{\vee}_\la$ were defined by (\ref{Center}), (\ref{hla}), 
respectively, $|W|$ is the number of elements of the Weyl group $W$ and $n$ is the rank of $G$.
\end{thm}
\begin{proof}
Since $\phi^s_\lambda$ vanishes on $F_M\setminus F^s_M$, we have
\begin{equation*}
\sca{\phi^s_\la}{\phi^s_{\la'}}_{F_M^s}=\sum_{x\in F^s_M}\ep(x)\phi^s_\la(x)\overline{\phi^s_{\la'}(x)}=\sum_{x\in F_M}\ep(x)\phi^s_\la(x)\overline{\phi^s_{\la'}(x)}.
\end{equation*}
The equality
\begin{equation*}
\sum_{x\in F_M} \ep(x) \phi^s_\la(x)\overline{\phi^s_{\la'}(x)}= \sum_{y \in \frac{1}{M}P^{\vee}/Q^{\vee}} \phi^s_\la(y)\overline{\phi^s_{\la'}(y)}
\end{equation*}
follows from (\ref{rfun2}) and (\ref{WFM}) and the $W-$invariance of the expression $\phi^s_\la(x)\overline{\phi^s_{\la'}(x)}$. Then, using the $W-$invariance of $\frac{1}{M}P^{\vee}/Q^{\vee}$ and \eqref{bdis}, we have
\begin{align*}
\sca{\phi^s_\la}{\phi^s_{\la'}}_{F_M^s} = &\sum_{w'\in W}\sum_{w\in W} \sum_{y\in \frac{1}{M}P^{\vee}/Q^{\vee}}\sigma^s(ww')e^{2\pi\i\sca{w\la-w'\la'}{y}}=\abs{W}\sum_{w'\in W}\sum_{y \in \frac{1}{M}P^{\vee}/Q^{\vee}}\sigma^s(w')e^{2\pi\i\sca{\la-w'\la'}{y}}\\ = & c\abs{W}M^n \sum_{w'\in W} \sigma^s(w')\delta_{w'\la',\la}.
\end{align*}
If $\la=w'\la'$, then we have from (\ref{lrfun2}) that $w'\la=\la=\la'$, i.e. $w'\in\mathrm{Stab}^\vee (\la)$.
Any $\la\in \Lambda_M^s$ is of the form $\la=b+MQ$ with $b\in MF^{s\vee}$. Then, considering \eqref{rfunstab2} and \eqref{FsFldual}, we have
$$\sigma^s (\mathrm{Stab}^\vee (\la))= \sigma^s \circ \widehat\psi \left(\mathrm{Stab}_{\widehat{W}^{\mathrm{aff}}}(b/M)\right)=\{1\} $$
i.e. we obtain $\sigma^s(w')=1$ for any $w'\in\mathrm{Stab}^\vee (\la)$, and consequently
$$\sum_{w'\in W} \sigma^s(w')\delta_{w'\la',\la}=\sum_{w'\in W}\delta_{w'\la',\la}=h^\vee_\la\delta_{\la',\la}.$$
The case of $S^l-$functions is similar.
\end{proof}

\begin{example}
The highest root $\xi$ and the highest dual root $\eta$ of $C_2$ are given by the formulas
$$\xi=2\al_1+\al_2,\ \ \eta=\al^\vee_1+2\al^\vee_2.$$
The Weyl group of $C_2$ has eight elements, $|W|=8$, and we calculate the determinant of the Cartan matrix 
$c=2$.
For a parameter with coordinates in $\omega-$basis $(a,b)$ and for a point with coordinates in $\alpha^\vee$-basis $(x,y)$, we have the following explicit form of $S^{s}-$ and $S^{l}-$functions of $C_2$:
\begin{align*}
\phi^s_{(a,b)}(x,y)= & 2\{ \cos(2\pi((a+2b)x-by))+\cos(2\pi(ax+by)) \\ &-\cos(2\pi((a+2b)x-(a+b)y))-\cos(2\pi(ax-(a+b)y))\}\\
\phi^l_{(a,b)}(x,y)= & 2\{ -\cos(2\pi((a+2b)x-by))+\cos(2\pi(ax+by)) \\ &-\cos(2\pi((a+2b)x-(a+b)y))+\cos(2\pi(ax-(a+b)y))\}.
\end{align*}
The grids $F^s_M$ and $F^l_M$ are given by
\begin{align*}
F^s_M(C_2) =& \setb{\frac{u^s_1}{M}\om^{\vee}_1+\frac{u^s_2}{M}\om^{\vee}_2}{u^s_0,\,u^s_2\in \Z^{\geq 0},\, u^s_1\in\N,\, u^s_0+2u^s_1+u^s_2=M} \\
F^l_M(C_2) =& \setb{\frac{u^l_1}{M}\om^{\vee}_1+\frac{u^l_2}{M}\om^{\vee}_2}{u^l_0,\,u^l_2\in \N,\, u^l_1\in\Z^{\geq 0},\, u^l_0+2u^l_1+u^l_2=M}
\end{align*}
and the grids of weights $\Lambda^s_M$ and $\Lambda^l_M$ are determined by
\begin{align*}
\Lambda^s_M(C_2) =& \setb{t^s_1\om_1+t^s_2\om_2}{t^s_0,t^s_1\in\N,t^s_2\in \Z^{\geq 0},\, t^s_0+t^s_1+2t^s_2=M}\\
\Lambda^l_M(C_2) =& \setb{t^l_1\om_1+t^l_2\om_2}{t^l_0,t^l_1\in\Z^{\geq0},t^l_2\in \N,\, t^l_0+t^l_1+2t^l_2=M}.
\end{align*}

The discrete orthogonality relations of $S^s-$ and $S^l-$ functions of $C_2$, which hold for any two functions $\phi^s_\la$, $\phi^s_{\la'}$ labeled by $\la,\la'\in\Lambda^s_M(C_2)$, and $\phi^l_\la$, $\phi^l_{\la'}$ labeled by $\la,\la'\in\Lambda^l_M(C_2)$, are of the form (\ref{ortho}) and (\ref{orthol}), respectively. The calculation procedure of the coefficients $\ep(x)$, $h_\lambda ^{\vee}$, which appear in (\ref{scp}),(\ref{ortho}) and (\ref{orthol}), is detailed in \S 3.7 in \cite{HP}. The values of the coefficients $\ep(x)$, $h_\lambda ^{\vee}$ for $x\in F^s_M(C_2) $, $\la\in \Lambda^s_M(C_2)$  and for $x\in F^l_M(C_2) $, $\la\in \Lambda^l_M(C_2)$ are listed in Table \ref{F4}. We represent each point $x\in F^s_M(C_2)$ and each weight $\lambda^s\in \Lambda_M(C_2)$ by the coordinates $[u^s_0,u^s_1,u^s_2]$ and $[t^s_0,t^s_1,t^s_2]$. Similarly, we represent each point $x\in F^l_M(C_2)$ and each weight $\lambda^l\in \Lambda_M(C_2)$ by the coordinates $[u^l_0,u^l_1,u^l_2]$ and $[t^l_0,t^l_1,t^l_2]$.
\begin{table}[ht]
\begin{tabular}{c|c}
$x\in F^s_M(C_2)$ & $\ep (x)$  \\ \hline
$[u^s_0,u^s_1,u^s_2]$ & $8$  \\
$[0,u^s_1,u^s_2]$ & $4$   \\
$[u^s_0,u^s_1,0]$ & $4$   \\
$[0,u^s_1,0]$ & $2$  \\
\multicolumn{1}{c}{}&\\
$x\in F^l_M(C_2)$ & $\ep (x)$  \\ \hline
$[u^l_0,u^l_1,u^l_2]$ & $8$  \\
$[u^l_0,0,u^l_2]$ & $4$   \\
\multicolumn{1}{c}{}&\\
\multicolumn{1}{c}{}&\\
\end{tabular}\hspace{24pt}
\begin{tabular}{c|c}
$\lambda\in \Lambda^s_M(C_2)$  & $h_\lambda ^{\vee}$ \\ \hline
$[t^s_0,t^s_1,t^s_2]$ & $1$  \\
$[t^s_0,t^s_1,0]$ & $2$   \\
\multicolumn{1}{c}{}&\\
\multicolumn{1}{c}{}&\\
\multicolumn{1}{c}{}&\\
$\lambda\in \Lambda^l_M(C_2)$  & $h_\lambda ^{\vee}$ \\ \hline
$[t^l_0,t^l_1,t^l_2]$ & $1$  \\
$[0,t^l_1,t^l_2]$ & $2$   \\
$[t^l_0,0,t^l_2]$ & $2$  \\
$[0,0,t^l_2]$ & $8$  \\
\end{tabular}\\
\medskip
\caption{The coefficients $\ep(x)$ and $h_\lambda ^{\vee}$ of $C_2$. All variables $u^s_0,u^s_1,u^s_2, t^s_0,t^s_1,t^s_2$ and $u^l_0,u^l_1,u^l_2, t^l_0,t^l_1,t^l_2$ are assumed to be natural numbers.}\label{F4}
\end{table}
\end{example}

\subsection{Discrete $S^s-$ and $S^l-$transforms}\

Analogously to ordinary Fourier analysis, we define interpolating functions $I^s_M$ and $I^l_M$
\begin{align}
I^{s}_M(x):= \sum_{\la\in \Lambda_M^s} c^s_\la \phi^s_\la(x),\q\q I^{l}_M(x):= \sum_{\la\in \Lambda_M^l} c^l_\la \phi^l_\la(x),\q x\in \R^n \label{intc}
\end{align}
which are given in terms of expansion functions $\phi^s_\la$ and $\phi^l_\la$ and expansion coefficients $c^s_\la$, $c^l_\la$, whose values need to be determined. These interpolating functions can also be understood as finite cut-offs of infinite expansions.
Suppose we have some function $f$ sampled on the grid $F_M^s$ or $F_M^l$. The interpolation of $f$ consists in finding the coefficients $c^s_\la$ or $c^l_\la$ in the interpolating functions (\ref{intc}) such that
\begin{align}
\begin{alignedat}{2}\label{intcs}
I^s_M(x)=& f(x), \q x\in F_M^s \\
I^l_M(x)=& f(x), \q x\in F_M^l
\end{alignedat}
\end{align}
Relations \eqref{complete}
and \eqref{ortho}, \eqref{orthol} allow us to view the values $\phi^s_\la(x)$ with $x\in F_M^s$, $\la\in \Lambda_M^s$ and the values $\phi^l_\la(x)$ with $x\in F_M^l$, $\la\in \Lambda_M^l$ as elements of non-singular square matrices.
These invertible matrices coincide with the matrices of the linear systems (\ref{intcs}). Thus, the coefficients $c^s_\la$ and $c^l_\la$ can be uniquely determined. The formulas for calculation of $c^s_\la$ and $c^l_\la$, which we call discrete $S^s-$ and $S^l-$transforms, can obtained by means of calculation of standard Fourier coefficients
\begin{align}\label{trans}
\begin{alignedat}{2}
c^s_\la=& \frac{\sca{f}{\phi^s_\la}_{F_M^s}}{\sca{\phi^s_\la}{\phi^s_\la}_{F_M^s}}=(c\abs{W} M^n h^{\vee}_\la)^{-1}\sum_{x\in F_M^s}\ep(x) f(x)\overline{\phi^s_\la(x)}\\
c^l_\la=& \frac{\sca{f}{\phi^l_\la}_{F_M^l}}{\sca{\phi^l_\la}{\phi^l_\la}_{F_M^l}}=(c\abs{W} M^nh^{\vee}_\la)^{-1}\sum_{x\in F_M^l}\ep(x) f(x)\overline{\phi^l_\la(x)}
\end{alignedat}
\end{align}
and the corresponding Plancherel formulas also hold
\begin{align*}
\sum_{x\in F^s_M} \ep(x)\abs{f(x)}^2 = & c \abs{W} M^n \sum_{\la\in\Lambda^s_M}h^{\vee}_\la|c^s_\la|^2 \\
\sum_{x\in F^l_M} \ep(x)\abs{f(x)}^2 = & c \abs{W} M^n \sum_{\la\in\Lambda^l_M}h^{\vee}_\la|c^l_\la|^2.
\end{align*}

\section{Concluding Remarks}
\begin{itemize}
\item In view of the ever-increasing amount of digital data, practically the most valuable property of the orbit functions of $C-$, $S-$, $S^l-$ and $S^s-$families is their discrete orthogonality. The four families are distinguished most notably by their behavior at the boundary of $F$. The functions of $S^l-$ and $S^s-$families do not have an analog in one variable, i.e. rank 1 simple Lie group.

\item
The product of two $S^s-$functions or two $S^l-$functions with the same underlying Lie group and the same arguments $x\in\R^n$ but different dominant weights, say $\lambda$ and $\lambda'$, decomposes into the sum of $C-$functions:
\begin{equation*}
\phi^s_\la(x)\cdot \phi^s_{\la'}(x)=\sum_{w\in W}\sigma^s(w) \Phi_{\la+w\la'}(x),\q \phi^l_\la(x)\cdot \phi^l_{\la'}(x)=\sum_{w\in W}\sigma^l(w) \Phi_{\la+w\la'}(x).
\end{equation*}
where $\Phi_\la$ denotes the (normalized) $C-$function $\Phi_\la=\phi_\la^\id$.

\item
The present work raises the question under which conditions converge the functional series $\{I^s_M\}_{M=1}^\infty$, $\{I^l_M\}_{M=1}^\infty$ assigned to a function $f:F\map\Com$ by the relations (\ref{intc}), (\ref{trans}).

\item In addition to the $C-$ and $S-$functions, which are multidimensional generalizations of common cosine and sine functions, the $E-$functions generalizing the exponential functions \cite{KP3} is also defined \cite{P}. The $E-$functions also admit discrete orthogonality \cite{HP2,MP2}. For these 'standard' $E-$functions, the kernel of the homomorphism $\sigma^e$, given by \eqref{parity}, is crucial. It turns out that there are altogether six types of $E-$functions once the kernels of the sign homomorphisms $\sigma^s$ and $\sigma^l$ are included in the definition. So far, these six types have been studied in full detail only for rank 2 Lie groups \cite{HaHrPa}.

\item A general one-to-one link between the orbit functions and orthogonal polynomials in $n$ variables was pointed out in \cite{NPT}. Extensive literature exists about orthogonal polynomials, although most of it pertains to 2-variable polynomials. It cannot be assumed that our Lie group defined polynomials of rank two were not taken into account. For a greater number of variables, not all of the polynomials defined from the simple Lie groups have been noticed. Discrete orthogonality of the polynomials in more than one variable is outside the scope of traditional approaches.

\end{itemize}

\section*{Acknowledgments}

We gratefully acknowledge the support of this work by the Natural Sciences and Engineering Research Council of Canada and by the Doppler Institute of the Czech Technical University in Prague. JH is grateful for the hospitality extended to him at the Centre de recherches math\'ematiques, Universit\'e de Montr\'eal. JH gratefully acknowledges support by the Ministry of Education of Czech Republic (project MSM6840770039). JP expresses his gratitude for the hospitality of the Doppler Institute.


\begin{thebibliography}{99}

\bibitem{BB}
A. Bjorner, F. Brenti,
{\it Combinatorics of Coxeter groups,}
Graduate Texts in Mathematics, {\bf 231} (2005) Springer, New York.

\bibitem{HaHrPa}
L. H\'akov\'a, J. Hrivn\'ak, J. Patera, {\it Six types of $E-$functions of the Lie groups $O(5)$ and $G(2)$,} J. Phys.~A: Math. Theor. {\bf 45} (2012),  125201 arXiv:1202.5031.

\bibitem{HP}
J. Hrivn\'ak, J. Patera, {\it On discretization of tori of compact simple Lie groups,} J. Phys. A: Math. Theor.~{\bf 42} (2009) 385208, arXiv:0905.2395.

\bibitem{HP2}
J. Hrivn\'ak, J. Patera, {\it On $E-$discretization of tori of compact simple Lie groups,} J. Phys. A: Math. Theor.~{\bf 43} (2010) 165206, arXiv:0912.4194.

\bibitem{H2}
J.~E.~Humphreys, {\it Reflection groups and Coxeter groups}, Cambridge Studies in Advanced Mathematics, {\bf 29} (1990) Cambridge University Press, Cambridge.

\bibitem{Kac}
V.~Kac, {\it  Automorphisms of finite order of semi-simple Lie algebras},
Funct.~Anal. and~Appl., {\bf 3} (1969) 252--254.

\bibitem{KP1}
A. U. Klimyk, J. Patera, {\it Orbit functions,\/}  SIGMA (Symmetry,
Integrability and Geometry: Methods and Applications) {\bf 2} (2006), 006, 60.
pages, math-ph/0601037.

\bibitem{KP2}
A. U. Klimyk, J. Patera, {\it Antisymmetric orbit functions,\/}  SIGMA (Symmetry, Integrability and Geometry: Methods and Applications) {\bf 3} (2007), paper 023, 83 pages;  math-ph/0702040.

\bibitem{KP3}
A. U. Klimyk, J. Patera, {\it $E$-orbit functions,\/} SIGMA (Symmetry, Integrability and Geometry: Methods and Applications) {\bf 4} (2008), 002, 57 pages; arXiv:0801.0822.

\bibitem{MMP}
R.~V.~Moody, L. Motlochov\'a, and J. Patera, {\it New families of Weyl group orbit functions,\/}, arXiv:1202.4415.

\bibitem{MP3}
R.~V.~Moody, J.~Patera, {\it Characters of elements of finite order
in simple Lie groups,\/} SIAM J. on Algebraic and Discrete Methods
{\bf 5} (1984) 359-383.

\bibitem{MP1}
R.~V.~Moody, J.~Patera, {\it  Computation of character decompositions of class functions on compact semisimple Lie groups,\/} Mathematics of Computation
{\bf 48} (1987) 799--827.

\bibitem{MP2}
R.~V.~Moody and J.~Patera, {\it Orthogonality within the families of \ $C-$,
$S-$, and $E-$functions of any compact semisimple Lie group,\/} SIGMA
(Symmetry, Integrability and Geometry: Methods and Applications) {\bf 2} (2006) 076, 14 pages, math-ph/0611020.

\bibitem{NPT}
 M. Nesterenko, J. Patera, A. Tereszkiewicz {\it Orthogonal polynomials of compact simple Lie groups,\/} International Journal of Mathematics and Mathematical Sciences (2011) 969424.

\bibitem{P}
J. Patera, {\it Compact simple Lie groups and theirs $C-$, $S-$, and
$E$-transforms,\/} SIGMA (Symmetry, Integrability and Geometry: Methods and
Applications) {\bf 1} (2005) 025, 6 pages, math-ph/0512029.

\bibitem{Sz}
M. Szajewska, {\it Four types of special functions of $G_2$
and their discretization,\/} Integral Transform. Spec. Funct. {\bf 23} (6) (2012) 445-472, arXiv:1101.2502.

\bibitem{VO}
E. B.~Vinberg, A. L.~Onishchik, {\sl Lie groups and Lie algebras,\/}
Springer, New York, 1994.

\end{thebibliography}
\end{document}